\documentclass[twocolumn, mat]{webofc}
\usepackage[varg]{txfonts}   
\usepackage{amsmath,epsfig,amsthm1,amssymb,url,color}

\newcommand{\be}{\begin{eqnarray}}
\newcommand{\ee}{\end{eqnarray}}

\DeclareMathOperator{\supp}{supp}

\def\R{\mathbb{R}}

\newcommand{\ep}{\varepsilon}

\newcommand{\om}{\Omega}

\newcommand{\diag}{{\rm diag}\,}

\newcommand{\1}{{\bf 1}}

\newcommand{\cof}{{\rm cof}\,}

\def\diag{\mbox{\rm diag}\,}

\def\det{\mbox{\rm det}\,}

\def\mint{-\kern-1.00em\int}

\newtheorem{thm}{Theorem}

\woctitle{ESOMAT 2015}
\usepackage{graphicx}
\graphicspath{%
    {converted_graphics/}
    {C:/Users/John/Pictures/}
}
\begin{document}\color{black}
\title{Geometry of polycrystals and microstructure}
%
%

\author{John M.  Ball\inst{1}\fnsep\thanks{\email{ball@maths.ox.ac.uk
}} \and
        Carsten Carstensen\inst{2}\fnsep\thanks{\email{cc@math.hu-berlin.de}} 
}

\institute{Mathematical Institute, University of Oxford, Andrew Wiles Building, Radcliffe Observatory Quarter, Woodstock Road, Oxford
OX2 6GG, U.K.
\and
           Department of Mathematics, Humboldt-Universit\"{a}t zu
Berlin, Unter den Linden 6, D-10099 Berlin, FRG 
          }
\selectlanguage{english}

\abstract{%
We investigate the geometry of polycrystals, showing that for polycrystals formed of convex grains the interior grains are polyhedral, while for polycrystals with general grain geometry the set of triple points is small. Then we investigate possible martensitic morphologies resulting from intergrain contact. For cubic-to-tetragonal transformations we show that homogeneous zero-energy microstructures matching a pure dilatation on a grain boundary necessarily involve more than four deformation gradients. We discuss the relevance of this result for  observations  of microstructures involving second and third-order laminates in various materials. Finally we consider the more specialized situation of bicrystals formed from materials having two martensitic energy wells (such as for orthorhombic to monoclinic transformations), but without any restrictions on the possible microstructure, showing how a generalization of the Hadamard jump condition can be applied at the intergrain boundary to show that a pure phase in either grain is impossible at minimum energy. 
}
\maketitle
\section{Introduction}
\label{intro}
In this paper we investigate the geometry of polycrystals and its implications for microstructure morphology within the nonlinear elasticity model of martensitic phase transformations \cite{j32,j40}. The rough idea is that the microstructure is heavily influenced by conditions of compatibility at grain boundaries resulting from continuity of the deformation.

 However, in order to express this precisely, it is first of all necessary  to give a careful mathematical description of the assumed grain geometry, something that is not often done even in mathematical treatments (a rare exception being \cite{taylor99}). In particular,  it is useful to be able to articulate the intuitively obvious fact that in the neighbourhood of most points of an interior grain boundary only two grains are present,  because it is at such points that it is easiest to apply compatibility conditions. 

A second issue is then to develop useful forms of the compatibilty conditions at such points, expressed in terms of deformation gradients, which on the one hand do not make unjustified assumptions about the microstructure morphology, and on the other hand can be exploited to draw conclusions about that morphology. 

The plan of the paper is as follows. In Section \ref{polycrystals} we give a precise description of grain geometry, defining interior and boundary grains, and the set of triple points. We then discuss the case of convex grains, showing that interior grains form convex polyhedra and that in 3D the set of triple points is a finite union of closed line segments. For possibly nonconvex grains we then show under weak conditions on the grain geometry  that  in 2D a polycrystal with $N$ grains can have at most $2(N-2)$ triple points, while  in arbitrary dimensions   the set of triple points is small. 

In Section \ref{micro} we address some examples in which compatibility at grain boundaries leads to restrictions on possible microstructures. First we show that, for a cubic-to-tetragonal transformation, a macroscopically homogeneous zero-energy microstructure matching a pure dilatation on the boundary must involve more than four values of the deformation gradient. We discuss the reasons why nevertheless second-order laminates, involving to a good approximation just four gradients in a single grain, are observed in materials, such as the ceramic ${\rm BaTiO}_3$ and RuNb alloys, which undergo cubic-to-tetragonal transformations. Then we consider the situation of a bicrystal with special geometry formed of a material undergoing a phase transformation with just two energy wells (such as cubic-to-orthorhombic), and without further assumptions on the microstructure give conditions under which a zero-energy microstructure must be complex, i.e. cannot be a pure variant in either grain; this analysis uses a generalization of the Hadamard jump condition developed in \cite{u5}. 

Finally in Section \ref{conclusion} we draw some conclusions and give some perspectives on possible future developments.
\section{Geometry of polycrystals}
\label{polycrystals}By a {\it domain} in $n$-dimensional Euclidean space $\R^n$, $n\geq 2$, we mean an open and connected subset of $\R^n$. (For the applications below   $n=2$ or $3$.) If $E\subset\R^n$ then $\overline E$ denotes the closure of $E$, $\partial E$ the boundary of $E$, and ${\rm int}\,E$ the interior of $E$. 
We consider a polycrystal which in a reference configuration occupies the bounded domain $\om\subset\R^n$. We suppose that the polycrystal is composed of a finite number  $N\geq 1$ of disjoint grains $\om_j, 1\leq j\leq N$, where each $\om_j$ is a bounded domain, so that  
\begin{equation}
\label{polyc}\om={\rm int}\,\bigcup_{j=1}^N\overline \om_j. 
\end{equation}
In general (see Theorem \ref{convexgrains} below) we cannot assume that the boundaries of the $\om_j$ are smooth. We will make various different assumptions concerning this below, but we always assume the minimal requirement  that each $\om_j$ is a {\it regular} open set, that is $\om_j={\rm int}\,\overline\om_j$. This avoids pathologies such as a grain consisting of an open ball with a single point at its centre removed. We can divide the grains into {\it interior grains} for which $\partial\om_j\subset\bigcup _{k\neq j}\partial\om_k$, and {\it boundary grains}, for which $\partial\om_j\setminus \bigcup _{k\neq j}\partial\om_k$ is nonempty. Note that an interior grain can have points of its boundary lying in $\partial\om$ (see Fig.~\ref{grainspicture}). We denote by $D=\bigcup_{j=1}^N\partial\om_j$ the union of the grain boundaries, and by $$T=\bigcup_{1\leq i_1<i_2<i_3\leq N}\partial\om_{i_1}\cap  \partial\om_{i_2}\cap\partial\om_{i_3}$$
the set of {\it triple points}, i.e. points which belong to the boundaries of three or more grains.
\begin{figure}[hbt]
\centerline{\includegraphics[width=3.83in,height=2.15in,keepaspectratio]{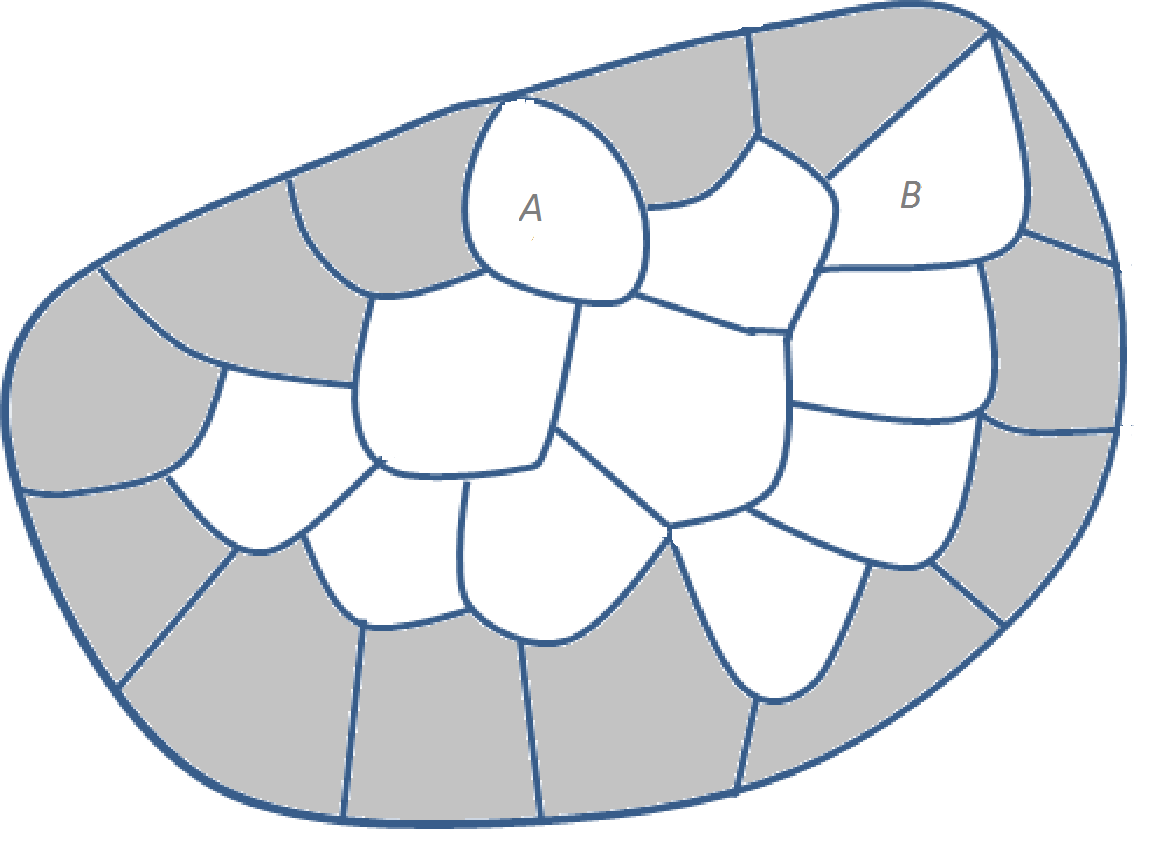}}
\caption{Schematic polycrystal grain structure in 2D, with boundary grains shaded. The boundaries of the two interior grains $A$ and $B$ have points in common with the boundary $\partial\om$ of the polycrystal. }\label{grainspicture}
\end{figure}
\subsection{Convex grains}
\label{convex}
We recall that a set $E\subset \R^n$ is said to be {\it convex} if the straight line segment joining any two points $x_1, x_2\in E$ lies in $E$, i.e. $\lambda x_1+(1-\lambda)x_2\in E$ for all $\lambda\in [0,1]$. An {\it open half-space} is a subset $H$ of $\R^n$ of the form $H=\{x\in\R^n: x\cdot e<k\}$ for some unit vector $e\in\R^n$ and constant $k$. A nonempty bounded open subset $P\subset \R^n$ is a {\it convex polyhedron} if $P$ is the intersection of a finite number of open half-spaces. The {\it dimension} of a convex set $E\subset\R^n$ is the dimension of the affine subspace of $\R^n$ spanned by $E$. The following statements, which are probably known, are elementary consequences of the hyperplane separation theorem (see, for example, \cite[Theorem 11.3]{rockafellar70}), which asserts that if  subsets $E,F$ are disjoint convex subsets of $\R^n$ with $E$ open, then there exist a unit vector $e\in\R^n$ and a constant $k$ such that 
$$x\cdot e<k\leq y\cdot e \mbox{ for all } x\in E, y\in F.$$
In particular, taking $F=\{z\}$ with $z\in\partial E$, any open convex set is regular. 
\begin{thm}
\label{convexgrains}
Suppose that each grain $\om_j$ is convex. Then\\
(i) each $\om_j$ is the intersection of $\om$ with a finite number of open half-spaces, \\
(ii) each interior grain is a convex polyhedron,\\
(iii) the set $T$ of triple points is a finite union of closed convex sets of dimension less than or equal to $n-2$.
\end{thm}
\begin{proof}  If $N=1$ there is nothing to prove, so we suppose that $N>1$.
By the hyperplane separation theorem, given a grain $\om_j$, for any $k\neq j$ there exists an open half-space $H_{j,k}$ such that $\om_j\subset H_{j,k}$ and $\overline\om_k\subset \R^n\setminus H_{j,k}$. Hence 
$$\om_j\subset \om\cap\bigcap_{k\neq j}H_{j,k}.$$ 
Let $x\in\om\cap\bigcap_{k\neq j}H_{j,k}$. Then since $\om\cap\bigcap_{k\neq j}H_{j,k}$ is open and disjoint from $\overline \om_k$ for $k\neq j$, it follows that $x$ is an interior point of $\overline\om_j$. Since $\om_j$ is regular, $x\in\om_j$, and hence $ \Omega_j = \Omega\cap
\bigcap_{k\neq j} H_{j,k}$. This proves (i).

Let $\om_j$ be an interior grain and suppose for contradiction that $x\in \bigcap_{k\neq j}H_{j,k}$ with $x\not\in\om_j$. Given any $x_0\in\om_j$ there exists a convex combination $y=\lambda x_0+(1-\lambda)x$, $\lambda\in[0,1]$, with $y\in\partial\om_j$. Since $\bigcap_{k\neq j}H_{j,k}$ is convex, $y\in\bigcap_{k\neq j}H_{j,k}$, and thus  $y\not\in\partial\om_k$ for $k\neq j$, contradicting that $\om_j$ is an interior grain. This proves (ii).

Given $1\leq i_1<i_2<i_3\leq N$ the set $K=\partial\om_{i_1}\cap\partial\om_{i_2}\cap\partial\om_{i_3}=\overline\om_{i_1}\cap\overline\om_{i_2}\cap\overline\om_{i_3}$ is closed and convex. Let $A$ denote the linear span of $K$. Then by \cite[Theorem 6.2]{rockafellar70}  there exists   a relative interior point $\bar x $ of $K$ in $A$,  that is for some $\ep>0$ the closed ball $\overline {B(\bar x,\ep)}=\{x\in\R^n:|x-\bar x|\leq\ep\}$ is such that $\overline {B(\bar x,\ep)}\cap A\subset K$. In particular the dimension of $K$, which by definition is the dimension of $A$, is less than $n$. Suppose for contradiction that the dimension of $K$ is $n-1$, so that $A=\{x\in\R^n:x\cdot e=k\}$ is a hyperplane. Then there exists a point $x_1\in \om_{i_1}$ which lies strictly on one side of $A$, say $x_1\cdot e<k$. Hence the closed convex hull of $x_1$ and $\overline {B(\bar x,\ep)}\cap K$ lies in $\overline\om_{i_1}$, and its interior contains the open half-ball $\{x\in\R^n:x\cdot e<k, |x-\bar x|<\ep'\}$ for some small $\ep'>0$. Since $\om_{i_1}$ is regular this half-ball is  a subset of $\om_{i_1}$. Repeating this argument for $i_2$ and $i_3$ we find a half-ball centre $\bar x$ which is a subset of two of the disjoint grains $\om_{i_1}, \om_{i_2},\om_{i_3}$. This contradiction implies (iii).
\end{proof}
\noindent Part (iii) implies that if $n=2$ there are  finitely many triple points (see Theorem \ref{triplepointsbound} below for a more general statement), while if $n=3$ then $T$ is the union of finitely many closed line segments. 
\subsection{Triple points in 2D}
\label{2D}
A famous counterexample in topology, the Lakes of Wada (see, for example, \cite{hockingyoung2nd, gelbaumolmsted}), shows that there can be three (or more) simply-connected, regular, open subsets of the closed unit square $[0,1]^2$ in $\R^2$ having a common boundary. Thus there is no hope to prove that the set $T$ of triple points is finite for $n=2$ without imposing further restrictions on the geometry of the grains $\om_j$. 
We will assume that each grain is a   bounded domain in $\R^2$ which is the region inside a Jordan curve, that is a non self-intersecting continuous loop in the plane. Such curves can be highly irregular. Nevertheless we can give a precise bound on the number of triple points.
\begin{thm}[\cite{u5}]
\label{triplepointsbound}
Assume that each grain $\om_j, j=1,\ldots,N$, is the region inside a Jordan curve. Then there are a finite number $m$ of triple points, and $m\leq 2(N-2)$.
\end{thm}
The bound is optimal, and attained for the configuration shown in Fig.~\ref{bound}.
 \begin{figure}[hbt]
\centerline{\includegraphics[width=3.07in,height=1.72in,keepaspectratio]{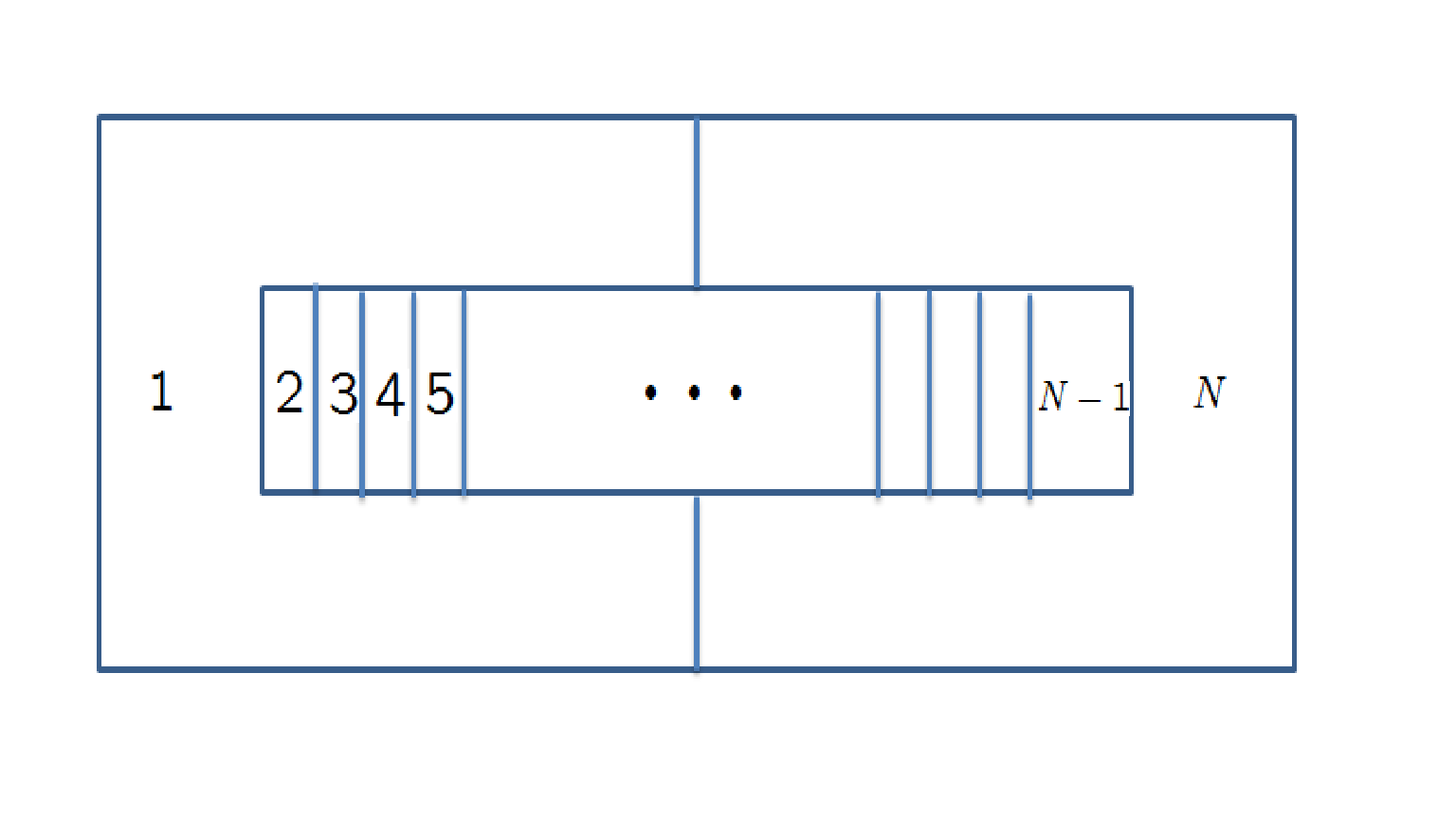}}
\caption{$N$ grains (labelled $1$ to $N$) in 2D with $2(N-2)$ triple points.}\label{bound}
\end{figure}
 The proof of Theorem \ref{triplepointsbound} involves a reduction to a problem of graph theory, as in the proof of the Four Colour Theorem for maps \cite{wilson14}, and use of Euler's formula relating the numbers of faces, vertices and edges of a polyhedron.
\subsection{Triple points in 3D}
\label{3D}
For dimensions $n=3$ and higher, we do not have as precise results as Theorem \ref{triplepointsbound}. However, we can prove under rather general conditions that the set $T$ of triple points is in some sense very small compared to the union of the grain boundaries $D$. We assume that the closure $\overline\om_j$ of each grain is a topological manifold with boundary, that is for each $x\in\bar\om_j$ there is a relatively  open neighbourhood $U(x)$ and a homeomorphism $\varphi$ between $U(x)$ and a relatively open neighbourhood of the closed half-space
$$\R^n_+:=\{(x_1,\ldots,x_n):x\cdot e_n\geq 0\},
$$
where $e_n=(0,\ldots,0,1)$. The precise details of this definition are not so important for this paper, but we note that if $n=2$ and $\om_j$ is the region inside  a Jordan curve then $\overline\om_j$ is a topological manifold with boundary, while for   $n\geq 2$ any domain whose boundary can be locally represented in suitable Cartesian coordinates by the graph of a continuous function,  and which lies on one side of its boundary, is also a topological manifold with boundary. Thus any geometry that is likely to be encountered in practice satisfies this condition. 
\begin{thm}[\cite{u5}]
\label{triplepointsnd}
Suppose that the closure $\overline\om_j$ of each grain is a topological manifold with boundary. Then the set 
$T$ of triple points is closed and nowhere dense in the union $D$ of grain boundaries, i.e. there is no point $x\in T$ and $\varepsilon>0$ such that $B(x,\ep)\cap D\subset T$.
\end{thm}
\noindent 
If $n=3$, then under the hypotheses of Theorem \ref{triplepointsnd} the set $T$ can have infinite length (technically, its one-dimensional Hausdorff measure can be infinite). One can conjecture that this is impossible if the grains $\om_j$ have more regular, for example Lipschitz, boundaries.
\section{Microstructure of polycrystals}
\label{micro}
In this section we derive some results concerning martensitic microstructure in polycrystals using the framework of the nonlinear elasticity model for martensitic phase transformations (see \cite{j32,j40}), in which at a constant temperature the total elastic free energy is assumed to have  the form 
\begin{equation}
\label{total}
I(y)=\int_\om W(x,\nabla y(x))\,dx,
\end{equation}
where $y:\om\to\R^3$ is the deformation, and $\om\subset\R^3$ has the form \eqref{polyc}, where we make the very mild additional assumption that the boundary $\partial\om_j$ of each grain has zero 3D measure (volume).   Denoting $M^{3\times 3}=\{\mbox{real }3\times 3\mbox{ matrices}\}$, $M^{3\times 3}_+=\{A\in M^{3\times 3}: \det A>0\}$ and $SO(3)=\{R\in M^{3\times 3}_+: R^TR=\1\}$, we suppose that the free-energy density $W$ is given by $W(x, A)=\psi(AR_j)$ for $x\in\om_j$, where $R_j\in SO(3)$ and $\psi$ is the free-energy density corresponding to a single crystal.  
We assume that $\psi:M^{3\times 3}_+\to [0,\infty)$ is continuous, frame-indifferent, that is
 \begin{equation}\label{frame-indifferent}
\psi(QA)=\psi(A)\mbox{  for all } A\in M^{3\times 3}_+, Q\in SO(3),
\end{equation}
 and has a symmetry group $\mathcal S$, a subgroup of $SO(3)$, so that
\begin{equation}
\label{cubic}
\psi(AR)=\psi(A)\mbox{  for all } A\in M^{3\times 3}_+, R\in {\mathcal S}.
\end{equation}
 For the case of cubic symmetry ${\mathcal S}= P^{24}$, the group of rotations of a cube into itself. We assume that we are working at a temperature at which the free energy of the martensite  (which we take to be zero) is less than that of the austenite, so that $K=\{A\in M^{3\times 3}_+: \psi(A)=0\}$ is given by 
\begin{equation}\label{energywells}
K=\bigcup_{i=1}^M SO(3) U_i,
\end{equation}
where the $U_i$ are positive definite symmetric matrices representing the different variants of martensite, so that the $U_i$ are the distinct matrices $RU_1R^T$ for $R\in P^{24}$.  

Zero-energy microstructures are represented by {\it gradient Young measures} $(\nu_x)_{x\in\om}$ satisfying $\supp\nu_x\subset KR_j^T$ for $x\in \om_j$. For each $x\in\om$, $\nu_x$ is a probability measure on $M^{3\times 3}$ that describes the asymptotic distribution of the deformation gradients $\nabla y^{(j)}$ of a minimizing sequence $y^{(j)}$ for $I$ (i.e. such that $I(y^{(j)})\to 0$) in a vanishingly small ball centred at $x$. Here the {\it support} $\supp \nu_x$ of $\nu_x$ is defined to be the smallest closed subset $E\subset M^{3\times 3}$ whose complement $E^c$ has zero measure, i.e. $\nu_x(E^c)=0$; intuitively $\supp \nu_x$  can be thought of as the limiting set of gradients at $x$. Thus the condition that $\supp\nu_x\subset KR_j^T$ for $x\in \om_j$ is equivalent to $\int_\om\int_{M^{3\times 3}}W(x,A)\,d\nu_x(A)\,dx=\sum_{j=1}^N\int_{\om_j}\int_{M^{3\times 3}}\psi(AR_j)\,d\nu_x(A)\,dx=0$ and expresses that the microstructure has zero energy. The corresponding macroscopic deformation gradient is given by $\nabla y(x)=\bar\nu_x=\int_{M^{3\times 3}}A\,d\nu_x(A)$. We note the {\it minors relations} 
\begin{eqnarray}\label{minors} \det \bar\nu_x=\langle \nu_x,\det\rangle=\int_{M^{3\times 3}}\det A\,d\nu_x(A),\\ \cof \bar\nu_x=\langle \nu_x,\cof\rangle=\int_{M^{3\times 3}}\cof A\, d\nu_x(A),\label{minors1}
\end{eqnarray}
where $\cof A$ denotes the matrix of cofactors of $A$. Note that \eqref{minors} implies that $\det\bar\nu_x=\det U_1$ for any zero-energy microstructure.
(See \cite{j56} for a description of gradient Young measures in the context of the nonlinear elasticity model for martensite.)

 In the case of cubic symmetry, and the absence of boundary conditions on $\partial\om$, there always exist such zero-energy microstructures. Indeed by the self-accommodation result of Bhattacharya \cite{bhattacharya92} for cubic austenite there exists a homogeneous gradient Young measure $\nu$ with $\supp \nu\subset K$ and $\bar\nu=(\det U_1)^\frac{1}{3}\1$. We can then define  for $x\in\om_j$ the measure  $\nu_x(E)=\nu(R_j^TER_j)$ of a subset $E\subset M^{3\times 3}$ of matrices. Then, since $R_j^TM^{3\times 3}R_j=M^{3\times 3}$ we have that $\bar\nu_x=\int_{M^{3\times 3}}A\,d\nu_x(A)=R_j\int_{M^{3\times 3}}B\,d\nu(B)R^T_j=(\det U_1)^\frac{1}{3}\1$ for $x\in\om_j$. By a result of Kinderlehrer \& Pedregal \cite{kinderlehrerpedregal91} it follows that $(\nu_x)_{x\in\om}$ is a gradient Young measure, and since $\int_{M^{3\times 3}}\psi(AR_j)\,d\nu_x(A)=\int_{M^{3\times 3}}\psi(R_jB)\,d\nu(B)=0$ it follows that $(\nu_x)_{x\in\om}$ is a zero-energy microstructure.
\subsection{Higher-order laminates for cubic-to-tetragonal transformations}
\label{fourgradients}
In this subsection we consider a cubic-to-tetragonal transformation, for which $K$ is given by \eqref{energywells} with $M=3$ and 
$U_1=\diag(\eta_2,\eta_1,\eta_1)$, $U_2=\diag(\eta_1,\eta_2,\eta_1)$, $U_3=\diag(\eta_1,\eta_1,\eta_2)$, where $\eta_1>0$, $\eta_2>0$ and $\eta_1\neq\eta_2$. Motivated by the observation above that a zero-energy microstructure with uniform macroscopic deformation gradient $(\det U_1)^\frac{1}{3}\1=\eta_1^\frac{2}{3}\eta_2^\frac{1}{3}\1$ exists for any polycrystal, we discuss whether this can be achieved with a microstructure that in each grain involves just $k$ gradients, where $k$ is small. Without loss of generality we can consider a single unrotated grain, so that the question reduces to whether there exists a homogeneous gradient Young measure $\nu$ having the form 
\begin{equation}\label{finite}
\nu=\sum_{i=1}^k\lambda_i\delta_{A_i} \mbox{ with }\lambda_i\geq 0, \sum_{j=1}^k\lambda_j=1, \mbox{ and } A_i\in K,
\end{equation}
and with macroscopic deformation gradient $\bar\nu = \eta_1^\frac{2}{3}\eta_2^\frac{1}{3}\1$. In \eqref{finite} we have used the notation $\delta_A$ for the Dirac mass at $A\in M^{3\times 3}$, namely the measure defined by 
$$\delta_A(E)=\left\{\begin{array}{ll}1& \mbox{if } A\in E,\\ 0&\mbox{if }A\not\in E.\end{array}\right.$$
 The following result implies in particular that this is impossible unless $k>4$, so that \eqref{finite} cannot be satisfied for a double laminate, a result also obtained by Muehlemann \cite{muehlemann15}.
\begin{thm}
\label{4gradthm}
There is no  homogeneous gradient Young measure $\nu$ with $\supp\nu\subset K=\cup_{i=1}^3SO(3)U_i$ and satisfying
  $\bar\nu=\eta_1^\frac{2}{3}\eta_2^\frac{1}{3}\1$, such that $\supp \nu \cap(SO(3)U_j\cup SO(3)U_k)$ contains at most two matrices for some distinct pair $j,k\in\{1,2,3\}$.
\end{thm}
\begin{proof} Suppose first that $\supp \nu$ is contained in the union of two of the wells, say $\supp\nu\subset SO(3)U_1\cup SO(3)U_2$. Then by the characterization of the quasiconvex hull of $SO(3)U_1\cup SO(3)U_2$  in \cite{j40} we have that $\bar\nu^T\bar\nu e_3=\eta_1^\frac{4}{3}\eta_2^\frac{2}{3}e_3=\eta_1^2 e_3$. Hence $\eta_1=\eta_2$, a contradiction. 
Without loss of generality we can therefore suppose that
\begin{equation}
\label{nuform}
\nu=\lambda_1 \mu +\lambda_2\delta_{R_2U_2}+\lambda_3\delta_{R_3U_3}
\end{equation}
 where $\lambda_i\geq 0$, $\sum_{i=1}^3\lambda_i=1$, $R_2, R_3\in SO(3)$ and $\mu$ is a probability measure on $SO(3)U_1$. Define $\mu^*(E)=\mu(EU_1)$ for $E\subset M^{3\times 3}$. Then $\mu^*$ is a probability measure with $\supp \mu^*\subset SO(3)$. Let $H=\bar\mu^*$. Then $\bar\mu=\int_{SO(3)U_1}A\,d\mu(A)=\int_{SO(3)}RU_1\,d\mu^*(R)=HU_1$. Letting $k=\eta_2/\eta_1$ and calculating $\bar\nu$ from \eqref{nuform}, we deduce that 
\begin{eqnarray}
\nonumber
\hspace{-.2in}k^\frac{1}{3}\1=\lambda_1 H \diag (k,1,1)+\lambda_2R_2\diag(1,k,1)&&\\
&&\hspace{-.95in}+\lambda_3R_3\diag(1,1,k).\label{mubar}
\end{eqnarray}
We now apply the minors relation \eqref{minors1} to $\nu$. Noting that 
\begin{eqnarray*}\langle\mu,\cof\rangle&=&\int_{SO(3)U_1}\cof A\,d\mu(A)\\&=&\int_{SO(3)}\cof(RU_1)\,d\mu^*(R)\\&=&\int_{SO(3)}R\,\cof(U_1)\,d\mu^*(R)\\&=&H\,\cof U_1,
\end{eqnarray*}
 we obtain
\begin{eqnarray}\nonumber\hspace{-.2in}
k^{-\frac{1}{3}}\1=\lambda_1 H \diag (k^{-1},1,1)+\lambda_2R_2\diag(1,k^{-1},1)&&\\
&&\hspace{-1.35in}+\lambda_3R_3\diag(1,1,k^{-1}).\label{cof}
\end{eqnarray}
Subtracting \eqref{cof} from \eqref{mubar} and dividing by $k-k^{-1}$ we deduce that
\begin{equation}\label{ck}
c(k)\1=\lambda_1He_1\otimes e_1+\lambda_2R_2e_2\otimes e_2+\lambda_3 R_3e_3\otimes e_3,
\end{equation}
where $c(k)=\frac{k^\frac{1}{3}-k^{-\frac{1}{3}}}{k-k^{-1}}=(1+k^\frac{2}{3}+k^{-\frac{2}{3}})^{-1}>0$, from which it follows that 
$$\lambda_1He_1=c(k)e_1, \; \lambda_2R_2e_2=c(k)e_2, \;\lambda_3R_3=c(k)e_3.$$
Hence $\lambda_2=\lambda_3=c(k)$. Acting \eqref{mubar} on $e_1$ we have that 
$$k^\frac{1}{3}e_1=c(k)(ke_1+R_2e_1+R_3e_1). $$ Hence $k^\frac{1}{3}\leq c(k)(k+2)$, from which we obtain $(k^\frac{1}{3}-k^{-\frac{1}{3}})^2\leq 0$, so that $k=1$, a contradiction.
\end{proof}
\noindent The conclusion of Theorem \ref{4gradthm} contrasts with observations of polycrystalline materials undergoing cubic-to-tetragonal phase transformations, but for which some grains are completely filled by a single double laminate. Such cases arise for the ceramic ${\rm BaTiO}_3$
 \cite{Arlt90}
 and in various RuNb and RuTa shape-memory alloys
 \cite{manzonietal09,
manzoni11,
vermautetal13,manzonietal14}.
 Arlt \cite{Arlt90}
 gives an interesting qualitative discussion of energetically preferred grain microstructure, drawing a distinction between the microstructures in interior and boundary grains. Following his reasoning, a likely explanation for why double, and not higher-order, laminates are observed in interior grains in these materials is that it is energetically better to form a double laminate with gradients   away from the energy wells, than to form a higher-order laminate having gradients extremely close to the energy wells. According to this explanation, the extra interfacial energy (ignored in the nonlinear elasticity model) involved in forming a higher-order laminate would exceed the total bulk plus interfacial energy for the double laminate. Of course once the gradients are allowed to move away from the wells the conclusion of Theorem \ref{4gradthm} will not hold. Additional factors could include cooperative deformation of different grains (so that the assumption of a pure dilatation on the boundary is not a good approximation), some of which may have more complicated microstructures than a  double laminate. It is interesting that third-order laminates are observed for RuNb alloys undergoing cubic-to-monoclinic transformations \cite{vermautetal13}.
\subsection{Bicrystals with two martensitic energy wells}
\label{bicrystals}
We now consider restrictions on possible zero-energy microstructures in polycrystals without making any assumptions other than those given by the grain geometry and texture. In particular, unlike in Section \ref{fourgradients}, we make no assumptions on the macroscopic deformation gradient of the microstructure. The restrictions result only from continuity of the deformation across grain boundaries. In order to give precise results, we restrict attention to a {\it bicrystal}, that is a polycrystal with just two grains $\om_1$ and $\om_2$. We assume that the grains have the cylindrical form
$\om_1=\omega_1\times (0,d),\; \om_2=\omega_2\times (0,d)$, where $d>0$, and $\omega_1, \omega_2\subset\R^2$ are bounded   domains. We assume for simplicity that the boundaries $\partial\omega_1, \partial\omega_2$ are smooth and intersect nontrivially, so that   $\partial\omega_1\cap\partial\omega_2$ contains points in the interior $\omega$ of $\overline\omega_1\cup\overline\omega_2$. The  interface between the grains $\partial\Omega_1\cap\partial\om_2=(\partial\omega_1\cap\partial\omega_2)\times (0,d)$ thus contains points in a neighbourhood of which $\om_1$ and $\om_2$ are separated by a smooth surface having normal $n(\theta)=(\cos\theta,\sin\theta,0)$  in the $(x_1,x_2)$ plane. 

We consider a martensitic transformation with two energy wells  (for example, orthorhombic-to-monoclinic)  for which   $M=2$ and $K=SO(3)U_1\cup SO(3)U_2$, where $U_1=\diag (\eta_2,\eta_1,\eta_3)$, $U_2=\diag(\eta_1,\eta_2,\eta_3)$, where $\eta_1>0$, $\eta_2>0$, $\eta_1\neq\eta_2$  and $\eta_3>0$. We further suppose that $\om_1$ has cubic axes in the coordinate directions $e_1, e_2, e_3$, while in $\om_2$ the cubic axes are rotated through an angle $\alpha$ about $e_3$. Thus a zero-energy microstructure corresponds to a gradient Young measure $(\nu_x)_{x\in\om}$ such that 
$$\supp \nu_x\subset K \mbox{ for }x\in\om_1,\;\;\supp\nu_x\subset KR(\alpha) \mbox{ for }x\in\om_2,$$
where $R(\alpha)=\left(\begin{array}{lll} \cos\alpha&-\sin\alpha&0\\ \sin\alpha&\cos\alpha& 0\\
0&0&1\end{array}\right)$.  It can be shown that $KR(\alpha)=K$ if and only if $\alpha=\frac{r\pi}{2}$ for some integer $r$. Hence we assume that $\alpha\neq\frac{r\pi}{2}$.

We ask whether it is possible for there to be a zero-energy microstructure {\it which is a pure variant in one of the grains}, i.e. either for $i=1$ or $i=2$,   $\nu_x=\delta_{Q(x)U_j}$ for $x\in\om_i$ and some $j$, where $Q(x)\in SO(3)$.   Since $\om_i$  is connected,  a standard result  \cite{reshetnyak67} shows that $\nu_x=\delta_{Q(x)U_j}$ for $x\in\om_i$ implies that $Q(x)$ is smooth and hence \cite{shield73} is a constant rotation, so that $\nabla y(x)$ is constant in $\om_i$.  
\begin{thm}[\cite{u5}]
\label{planar}
Suppose that the  interface between the grains  is planar, i.e. $\partial\om_1\cap\partial\om_2\subset \Pi(N)$ where $\Pi(N)=\{x\in\R^3:x\cdot N= a\}$ for some unit vector $N=(N_1,N_2,0)$ and constant $a$. Then there exists  a zero-energy microstructure which is a pure variant in one of the grains. 
\end{thm}
\noindent Thus, in order to eliminate the possibility of a pure variant in one of the grains we need a curved  interface. To give an explicit result we consider the special case when $\alpha=\frac{\pi}{4}$. Let 
\begin{eqnarray*} 
&D_1=(\frac{\pi}{8},\frac{3\pi}{8})\cup (\frac{5\pi}{8},\frac{7\pi}{8})\cup(\frac{9\pi}{8},\frac{11\pi}{8})\cup(\frac{13\pi}{8},\frac{15\pi}{8}), \\   &D_2=(\frac{-\pi}{8},\frac{\pi}{8})\cup (\frac{3\pi}{8},\frac{5\pi}{8})\cup(\frac{7\pi}{8},\frac{9\pi}{8})\cup(\frac{11\pi}{8},\frac{13\pi}{8}). 
\end{eqnarray*}
\begin{thm}[\cite{u5}]
\label{45}
Let $\alpha=\frac{\pi}{4}$ and $\frac{\eta_2}{\eta_1}\leq\sqrt{1+\sqrt{2}}$. If $\partial\om_1\cap\partial\om_2$ has  points with normals $n(\theta)$  and $n(\theta')$ with $\theta\in D_1$ and $\theta'\in D_2$, then there is no zero-energy microstructure which is a pure variant in one of the grains.
\end{thm}
\noindent The main ingredients in the proofs of Theorems \ref{planar} and \ref{45} are (i) a reduction to two dimensions using the plane strain result \cite{j39}, (ii) the characterization in \cite{j40} of the quasiconvexification of $K$, (iii) a generalization of the Hadamard jump condition that implies that the difference between the polyconvex hulls of suitably defined limiting sets of gradients, on either side of a point on the grain boundary where the normal is $n$, contains a rank-one matrix $a\otimes n$, and (iv) long and detailed calculations.
\section{Conclusions and perspectives}
\label{conclusion}
In this paper we have provided a framework for discussing the effects of grain geometry on the microstructure of polycrystals as described by the nonlinear elasticity model of martensitic transformations. This consists of two threads, a description and analysis of the grain geometry itself, and the use of generalizations of the Hadamard jump condition and other techniques to delimit possible zero-energy microstructures compatible with a given grain geometry. 

Both threads need considerable development. The quantitative description of polycrystals, as described for example in the book \cite{kurzydlowskiralph}, is a large subject which has many aspects (for example, sectioning and stochastic descriptions) for which a more rigorous treatment would be valuable. 

The problem of  determining possible zero-energy microstructures  is essentially one of multi-dimensional calculus, namely that of determining deformations compatible with a given geometry having  deformation gradients lying in, or  Young measures supported in,   the energy-wells  corresponding to each grain. Nevertheless we are very far from understanding how to solve it in any generality, one   obstacle being the well-known lack of a useful characterization of quasiconvexity (see, for example, \cite{p31}), which is known to be a key to understanding compatibility. The generalizations of the Hadamard jump conditions considered in \cite{u5} (see also \cite{iwaniecetal02}) are also insufficiently general and tractable. As well as for polycrystals, such generalized jump conditions are potentially relevant for the analysis of nonclassical austenite-martensite interfaces as proposed in \cite{j46,j48},   which have  been observed in CuAlNi \cite{j63,j64}, ultra-low hysteresis alloys \cite{song13}, and which have been suggested to be involved in steel \cite{koumatosmuehlemann15}. 

Despite the usefulness of the nonlinear elasticity theory, we have seen in connection with Theorem \ref{4gradthm} that in some situations the effects of interfacial energy can make its predictions of microstructure morphology inconsistent with experiment. This highlights the importance of developing a better understanding of how polycrystalline microstructure depends on the small parameters describing grain size and interfacial energy. 
\section*{Acknowledgements}  The research of JMB was supported by by the EC (TMR contract FMRX - CT  EU98-0229 and
ERBSCI**CT000670), by 
EPSRC
(GRlJ03466, the Science and Innovation award to the Oxford Centre for Nonlinear
PDE EP/E035027/1, and EP/J014494/1), the ERC under the EU's Seventh Framework Programme
(FP7/2007-2013) / ERC grant agreement no 291053 and
 by a Royal Society Wolfson Research Merit Award. We thank Philippe Vermaut for useful discussions concerning RuNb alloys.
\bibliography{gen2,balljourn,ballconfproc,ballprep}
\end{document}